\documentclass[aps,prl,reprint,groupedaddress]{revtex4-2}
\usepackage{dcolumn}
\usepackage{CJK,amsmath,amsthm,amssymb}
\usepackage{graphicx}
\usepackage{epstopdf}
\usepackage{epsfig}
\usepackage{subfigure}
\usepackage{caption2}
\usepackage[OT2,OT1]{fontenc}
\usepackage{hyperref}

\usepackage{mathrsfs,color,indentfirst,textcomp}
\def\qed{\leavevmode\unskip\penalty9999\hbox{}
         \nobreak\hfill
    \quad\hbox{\leavevmode  \hbox to.77778em{\hfil\vrule  \vbox to.675em {\hrule width.3em\vfil\hrule}\vrule\hfil}}
     \par\vskip 0pt}
\def\ra{\rangle}
\def\la{\langle}
\def\be{\begin{equation}}
\def\ee{\end{equation}}
\def\ba{\begin{array}}
\def\ea{\end{array}}

\newtheorem{theorem}{Theorem}

\newtheorem{lemma}{Lemma}

\begin{document}

\title{Unconditional robustness of multipartite entanglement of superposition}

\author{Hui-Hui Qin$^{1}$}

\author{Shao-Shuai Zhao$^{2}$}

\author{Shao-Ming Fei$^{3}$}

\affiliation{$^{1}$ Department of Sciences, Hangzhou Dianzi University, Hangzhou 310018, China\\
$^2$ Beijing National Day School, Beijing 100039, China\\
$^3$ School of Mathematical Sciences, Capital Normal University, Beijing 100048, China}



\begin{abstract}
We study the robustness of genuine multipartite entanglement and inseparability of multipartite pure states under superposition with product pure states. We introduce the concept of the maximal and the minimal Schmidt ranks for multipartite states. From the minimal Schmidt rank of the first order we present criterion of verifying unconditional robustness of genuine multipartite entanglement of multipartite pure states under superposition with product pure states. By the maximal Schmidt rank of the first order we verify the unconditional robustness of multipartite inseparability under superposition with product pure states. The number of product states superposed to a given entangled state which result in a separable state is investigated in detail. Furthermore, the minimal Schmidt ranks of the second order are also introduced to identify the unconditional robustness of an entangled state for tripartite inseparability.
\end{abstract}

\maketitle

\section{1. Introduction}
Quantum entanglement, especially genuine multipartite entanglement (GME) \cite{Horodecki} plays an important role in quantum information processing \cite{Braunstein} and quantum physics \cite{Vedral}, with significant applications in fields such as quantum teleportation \cite{Cavalcanti}, quantum dense coding \cite{Bruss} and quantum key distribution \cite{Barrett}. Entanglement is basically due to the superposition principle in quantum physics \cite{Nielsen}. However, superposition of an entangled pure state and product pure states may not always lead to entanglement. Here, the Schmidt rank of bipartite pure states \cite{Terhal,Nielsen} plays an important role both in the entanglement clarification of quantum states \cite{Sanpera,Sperling,Zhang} and in the purification of quantum systems \cite{Duer}. Halder and Sen showed in \cite{Halder} that any superposition of a bipartite entangled pure state and a product pure state produces only an entangled state if the initial entangled state has Schmidt rank no less than two, independent of the superposition coefficients and the choice of product states. The original entangled state is then called unconditionally robust in entanglement under superpositions, or unconditional inseparability under superpositions.

Since there are no longer universal Schmidt decompositions for multipartite quantum systems, the problem of unconditional robustness or unconditional inseparability under superpositions becomes more complicated. The superposition of a multipartite entangled state and a product state may result in both entanglement including GME and full separability. For example, a superposition of the maximally genuinely multipartite entangled state $|\psi_+\ra=\frac{1}{\sqrt{2}}(|000\ra+|111\ra)$ and $|000\ra$ can be fully separable $\sqrt{2}|\psi_+\ra-|000\ra=|111\ra$. Another superposition of $|\psi_+\ra$ and $|011\ra$ can be genuinely multipartite entangled $a|\psi_+\ra+b|011\ra$. 

In this study, we first introduce the concepts of the minimal and the maximal Schmidt rank of the first order for multipartite pure states. Based on these Schmidt ranks we obtain two criteria that identify whether a genuinely multipartite entangled state is unconditionally robust in GME and in inseparability, respectively in the second section. Also in the second section the corresponding Schmidt rank of the second order is also defined, which clarifies whether a superposition of a genuinely multipartite entangled pure state and a set of fully product states is triple separable. A series of detailed examples are given to illustrate all these conclusions. In the third section we consider the minimal number of superposed product states for vanishing the entanglement of a given bipartite pure state. One step further we also consider the number of superposed fully product states for vanishing the genuinely multipartite entanglement of a given multipartite pure state. In the last section some possible generalizations which might be considered are put forward.


\section{2. Unconditionally robust in GME and inseparability under superposition}
Let $\mathcal{H}_{A_i}$ be the Hilbert space associated with the $i$th subsystem $A_i~(i=1,2,\ldots,n)$ of the multipartite system $A=A_1A_2\ldots A_n$. We denote $\mathcal{G}_2(A)$ the set of all bi-partitions of $A$. We then define the minimal (maximal) Schmidt rank of the first order of an n-partite pure state $|\psi\ra\in \otimes^{n}_{i=1}\mathcal{H}_{A_i}$ as
\be
\begin{aligned}
r^1_{{\rm min}}\hat{=}\min_{X|\bar{X}\in\mathcal{G}_2(A)}
Rank[\rho_{X}],\\
r^1_{{\rm max}}\hat{=}\max_{X|\bar{X}\in\mathcal{G}_2(A)}
Rank[\rho_{X}],
\end{aligned}
\ee
where ${\rm Rank}[\rho_X]$ is the rank of the reduced density matrix $\rho_X=Tr_{X}|\psi\ra\la\psi|$ of $|\psi\ra$.

According to the above definition we have the following conclusion,
\begin{lemma}
An $n$-partite pure state $|\psi\ra\in\otimes^n_{i=1}\mathcal{H}_{A_i}$ is genuinely multipartite entangled if and only if $r^1_{{\rm min}}\geq2$, and entangled if and only if $r^1_{{\rm max}}\geq 2$.
\end{lemma}

\begin{proof}
It is obvious that $|\psi\ra\in\otimes^n_{i=1}\mathcal{H}_{A_i}$ is genuinely multipartite entangled if and only if there exists no bi-partition $X|\bar{X}\in \mathcal{G}_2(A)$ satisfying that $|\psi\ra=|\psi_X\ra\otimes|\psi_{\bar{X}}\ra$. That is to say ${\rm Rank}\Big[Tr_{X}|\psi\ra\la\psi|\Big]\geq 2$ for arbitrary bi-partition $X|\bar{X}\in\mathcal{G}_2(A)$. Equivalently, $r^1_{{\rm min}}\geq 2$, while $|\psi\ra$ is entangled if and only if there exists at least one bi-partition $Y|\bar{Y}\in\mathcal{G}_2(A)$ such that ${\rm Rank}[\rho_{Y}]\geq2$, i.e., $r^1_{{\rm max}}\geq {\rm Rank}[\rho_{Y}]\geq2$. 
\end{proof}

Unconditional inseparability of superposition is defined in \cite{Halder}. A bipartite entangled pure state $|\psi\ra$ is unconditionally robust in inseparability under superposition, if for arbitrary $p\in(0,1)$ and arbitrary product state $|p\ra=|\alpha\ra\otimes|\beta\ra$ such that the superposed state $\sqrt{p}|\psi\ra+\sqrt{1-p}|p\ra$ cannot be separable. For multipartite entangled pure states, inspired by \cite{Halder} we first give the following definition. If it is impossible to obtain a bi-separable state by superposing a genuinely multipartite entangled pure state $|\psi\ra$ and any product pure states, we call $|\psi\ra$ {\it unconditionally robust in GME} under superposition.
If $|\psi\ra$ is entangled, but not necessarily genuinely multipartite entangled, we say that $|\psi\ra$ is {\it unconditionally robust in inseparability} under superposition if no fully separable states can be obtained by superposing $|\psi\ra$ with any product states nontrivially. Obviously, if a multipartite pure state $|\psi\ra$ is unconditionally robust in GME under superposition, then it must be unconditionally robust in inseparability under superposition. According to our definition of Schmidt ranks of first order, we have

\begin{theorem}\label{theorem1}
Any nontrivial superposition of a $n$-partite genuinely multipartite entangled pure state $|\psi\ra$ with $r^1_{{\rm min}}\geq 3$ and arbitrary product pure states is still genuinely multipartite entangled; namely, arbitrary genuinely multipartite entangled pure state with $r^1_{{\rm min}}\geq 3$ is unconditionally robust in GME under superposition. More precisely, any nontrivial superposition of $|\psi\ra$ and $r^1_{min}-2$ product pure states gives rise to a genuinely multipartite entangled state.
\end{theorem}

\begin{proof}
Assume that under the bi-partition $Y|\bar{Y}$ of $|\psi\ra$ one gets the minimal Schmidt rank $r^1_{{\rm min}}$ of the first order, i.e., ${\rm Rank}[\rho_Y]={\rm Rank}[Tr_Y|\psi\ra\la\psi|]=r^1_{{\rm min}}$. Then $|\psi\ra$ can be decomposed into Schmidt form,
$$
|\psi\ra=\sum^{r^1_{{\rm min}}}_{i=1}
a_i|i_Y\ra|i_{\bar{Y}}\ra,$$
where $\{|i_Y\ra\}$ and $\{|i_{\bar{Y}}\ra\}$ are respectively orthonormal sets in subspaces $Y$ and $\bar{Y}$, $\sum_i|a_i|^2=1$. Consider the superposition of $|\psi\ra$ and an arbitrary product pure state, $|p\ra=\otimes^n_{i=1}|\alpha_i\ra=
\left(\underset{|\alpha_i\ra\in Y}{\otimes}|\alpha_i\ra\right)\otimes\left(\underset{|\alpha_k\ra\in \bar{Y}}{\otimes}|\alpha_k\ra\right)\hat{=}|\alpha_Y\ra\otimes|\alpha_{\bar{Y}}\ra$. According to Theorem 1 in \cite{Halder} any superposed state $|\psi'\ra=\lambda|\psi\ra+
\mu|p\ra=\lambda|\psi\ra+\mu|\alpha_Y\ra
\otimes|\alpha_{\bar{Y}}\ra$ with nontrivial superposition coefficient $\lambda\neq0$ satisfies that ${\rm Rank}[Tr_{Y}|\psi'\ra\la\psi'|]\geq r^1_{{\rm min}}-1\geq 2$. Then for arbitrary bi-partition $X|\bar{X}\in\mathcal{G}_2$, $|\psi\ra$ can be decomposed into another Schmidt form $|\psi\ra=\sum^r_{k=1}b_i|k_X\ra|k_{\bar{X}}\ra$, where $r={\rm Rank}[Tr_{X}|\psi\ra\la\psi|]\geq r^1_{{\rm min}}$ is the Schmidt rank in this bi-partition. Therefore, the superposed state $|\psi'\ra$ satisfies that ${\rm Rank}[Tr_{X}|\psi'\ra\la\psi'|]\geq r-1 \geq r^1_{{\rm min}}-1\geq 2$. This means that $|\psi'\ra$ is genuinely multipartite entangled, and $|\psi\ra$ is unconditionally robust in GME under superposition.

According to the generalization of Theorem 1 in \cite{Halder}, by superposing $|\psi\ra$ and $r^1_{{\rm min}}-2$ fully product pure states $\{|p_i\ra\}^{r^1_{{\rm min}}-2}_{i=1}$ we finally get a state $|\tilde{\psi}\ra=\lambda|\psi\ra+\sum^{r^1_{{\rm min}}-2}_{i=1}\mu_{i}|p_i\ra$ with nontrivial superposition coefficient $\lambda\neq0$ satisfying 
$$
{\rm Rank}[Tr_X|\tilde{\psi}\ra\la\tilde{\psi}|]\geq r^1_{{\rm min}}-(r^1_{{\rm min}}-2)=2
$$
for arbitrary bi-partition $X|\bar{X}\in\mathcal{G}_2$. It implies that $|\tilde{\psi}\ra$ is still a genuinely multipartite entangled state.
\end{proof}

\autoref{theorem1} characterizes a class of multipartite pure states that is unconditionally robust in GME under superposition. In addition to the unconditionally robust in GME under superposition, we have following conclusion for the unconditionally robust in inseparability under superposition.

\begin{theorem}\label{theorem2}
Any nontrivial superposition of an entangled pure state (may not be genuinely multipartite entangled) $|\psi\ra$ with $r^1_{{\rm max}}\geq 3$ and $r^1_{{\rm max}}-2$ product pure states gives rise to an entangled state.
\end{theorem}

\begin{proof}
The proof is similar to the proof of \autoref{theorem1}. Assume that under the bi-partition $Y|\bar{Y}$ the state $|\psi\ra$ attains the Schmidt Rank $r^1_{{\rm max}}$, i.e., $|\psi\ra$ can be decomposed into the form
$$
|\psi\ra=\sum^{r^1_{{\rm max}}}_{i=1}
c_i|i_{Y}\ra|i_{\bar{Y}}\ra,
$$
where $\sum^{r^1_{{\rm max}}}_{i=1}|c_i|^2=1$. According to the generalization of Theorem 1 in \cite{Halder} any nontrivial superposition of $|\psi\ra$ and a set of fully product states $\{|p_i\ra=|\alpha^i_1\ra\otimes|\alpha^i_2\ra\otimes\cdots\otimes|\alpha^i_n\ra=|\alpha^i_Y\ra\otimes
|\alpha^i_{\bar{Y}}\ra|\}^{r^1_{{\rm max}}-2}_{i=1}$ still gives rise to an entangled state under this bi-partition.
\end{proof}

The theorems above can be also understood as follows. For a genuinely multipartite entangled pure state $|\psi\ra$ with $r^1_{{\rm max}}\geq r^1_{{\rm min}}\geq 3$, one cannot obtain any bi-separable state by superposing $|\psi\ra$ with $r^1_{{\rm min}}-2$ product pure states according to \autoref{theorem1}. \autoref{theorem2} tells us that no fully separable state can be obtained by superposing $|\psi\ra$ with $r^1_{{\rm max}}-2$ product pure states. Next, we introduce the minimal Schmidt rank of second order for further investigations.

For an $n$-partite pure state $|\psi\ra\in\otimes^n_{i=1}\mathcal{H}_{A_i}$, under given bi-partition $X_1|X_2\in\mathcal{G}_2(A)$ there exists an unique $r(X_1|X_2)={\rm Rank}[Tr_{X_1}|\psi\ra\la\psi|]$ such that
$$
|\psi\ra=\sum^{r(X_1|X_2)}_{i_1=1}a_{i_1}|\psi_{i_1}(X_1)\ra|\phi_{i_1}(X_2)\ra
$$
for some orthonormal states $\{|\psi_{i_1}(X_1)\ra\}$ and $\{|\phi_{i_1}(X_2)\ra\}$ in subspaces $X_1$ and $X_2$, respectively. If $X_1$ contains two subspaces $X_{11}$ and $X_{12}$, one can sequentially partition $X_1$ as $X_{11}|X_{12}$. Then there exists a unique Schmidt rank for every $|\psi_{i_1}(X_1)\ra$ such that $$
|\psi_{i_1}(X_1)\ra=\sum^{r(X_{11}|X_{12};
X_1|X_2,i_1)}_{i_2=1}a_{i_2i_1}|\psi_{i_2i_1}(X_{11})\ra|\phi_{i_2i_1}(X_{12})\ra,
$$
where
$$r(X_{11}|X_{12};X_1|X_2,i_1)={\rm Rank}[Tr_{X_{11}}|\psi_{i_1}(X_1)\ra\la\psi_{i_1}(X_1)|]$$
is the Schmidt rank of $|\psi_{i_1}(X_1)\ra$. For the sake of convenience we define
\begin{widetext}
$$
r(X_{11}|X_{12};X_1|X_2,i_1)=\left.\Big\{
\begin{aligned}
&{\rm Rank}[Tr_{X_{11}}|\psi_{i_1}(X_1)\ra\la\psi_{i_1}(X_1)|], ~~~{\rm if}~X_{11}|X_{12}~ {\rm is}~{\rm an}~{\rm available}~{\rm bi-partition}~{\rm of}~\mathcal{G}_2(X_1);\\
&1,\qquad\qquad\qquad\qquad\qquad\qquad\qquad ~{\rm if}~ X_1 ~{\rm is}~{\rm single}~ {\rm subspace}.
\end{aligned}\right.
$$
We then define the minimal Schmidt rank of the second order of $|\psi\ra$ by
\be
r^2_{{\rm min}}=\min_{X_1|X_2\in\mathcal{G}_2(A)}\sum^{r(X_1|X_2)}_{i_1=1}\min_{X_{11}
|X_{12}\in\mathcal{G}_2(X_1)}r(X_{11}|X_{12};X_1|X_2,i_1).
\ee
\end{widetext}
Based on the Schmidt rank of the second order we have the following results about triple separable states.

\begin{theorem} 
Any nontrivial superposition of an $n$-partite pure state $|\psi\ra$ with $r^2_{{\rm min}}\geq 3$ and any $r^2_{{\rm min}}-2$ product pure states cannot be a triple separable state.
\end{theorem}

\begin{proof}
According to the definition of $r^2_{{\rm min}}$, it is obvious that $r^2_{{\rm min}}\geq r^1_{{\rm min}}$. The state $|\psi\ra$ is genuinely multipartite entangled when $r^1_{{\rm min}}\geq 2$. The Schmidt rank of second order takes over the minimal values of Schmidt rank of all kinds of bi-partitions. After superposing $r^2_{{\rm min}}-2$ product pure states with $|\psi\ra$, the resulted state can not be triple separable, as there exists at least one $|\psi_{i_1}(X_1)\ra$ which is genuinely multipartite entangled in subsystem $X_1$. Then the total state can not be triple separable.
\end{proof}

\bigskip
{\it Example 1.} Let us consider the generalized $n$-partite GHZ state, $|\psi\ra=\frac{1}{\sqrt{d}}\sum^{d-1}_{i=0}|ii\cdots i\ra\in \otimes^n_{i=1}\mathcal{H}_{A_i}$. We have that $r(X_1|X_2)=d$ for all bi-partitions in $\mathcal{G}_2(A)$, and $r(X_{11}|X_{12};X_1|X_2,i_1)=1$ for arbitrary $i_1$ and bipartition in $\mathcal{G}_{2}(X_1)$. Hence, $r^2_{{\rm min}}=r^1_{{\rm min}}=d$. Therefore, it is impossible to obtain a bi-separable state by superposing $|\psi\ra$ with $d-2$ product pure states. Of course it is impossible to obtain a triple separable state either.

\bigskip
Note that in Theorems 1 and 2, the entangled state $|\psi\ra$ and the arbitrary product states $|p\ra$ are not necessary to be orthogonal each other. One may also consider the orthogonal superposition of $|\psi\ra$ and product pure states in the complementary space of $|\psi\ra$. Let us consider the following examples in Hilbert space $\mathbb{C}^2\otimes\mathbb{C}^2\otimes\mathbb{C}^2$.

\bigskip
{\it Example 2.} We consider the genuinely multipartite entangled state $|\psi_1\ra=a_0|000\ra+a_1|111\ra$ with $a^2_0+a^2_1=1$. It is directly verified that $|\psi_1\ra$ is unconditionally robust in GME under superposition, with respect to any product pure states $|p\ra=|\alpha\ra|\beta\ra|\gamma\ra$ which is orthogonal to $|\psi_1\ra$, since any nontrivial superposed state $\cos\alpha|\psi\ra+\sin\alpha|p\ra$ $(\cos\alpha\sin\alpha\neq0)$ has all reduced density matrices of rank $2$.

\bigskip
{\it Example 3.} Consider  $|\psi_2\ra=\frac{1}{\sqrt{3}}(|000\ra+
|100\ra+|111\ra)$, which is a GMEd with ${\rm Rank}[\rho_i]=2$ for all the reduced single qubit states $(i=1,2,3)$. $|\psi_2\ra$ is not unconditionally robust in GME under superposition. This can be seen from that the superposed state $\frac{\sqrt{3}}{2}|\psi_2\ra+\frac{1}{2}|001\ra=|+\ra\otimes|\psi_+\ra$ is bi-separable, where $|+\ra=\frac{1}{\sqrt{2}}(|0\ra+|1\ra)$ and $|\psi_+\ra=\frac{1}{\sqrt{2}}(|00\ra+|11\ra)$.

\bigskip
{\it Example 4.} Consider state $|\psi_3\ra=\frac{\sqrt{2}}{4}|001\ra+
\frac{\sqrt{5}}{4}|01\ra|+\ra+\frac{3}{4}|1\ra|b\ra|+\ra$, where $|b\ra=\frac{2}{3}|0\ra+\frac{\sqrt{5}}{3}|1\ra$ and $|+\ra$ is defined in Example 3. $|\psi_3\ra$ is genuinely multipartite entangled with ${\rm Rank}[\rho_i]=2$ $(i=1,2,3)$. It can be verified that the superposition of $|\psi_3\ra$ and the orthogonal product state $|p\ra=|000\ra$ gives a separable state, $\frac{2\sqrt{2}}{3}|\psi_3\ra+\frac{1}{3}|p\ra
=|+\ra|b\ra|+\ra$.

\bigskip
The examples above show that superpositions of a genuinely multipartite entangled pure state and an orthogonal product pure state result in different types of entanglement or full separability, depending both on the superposition coefficients and the choices of the product pure states.
If more product states are superposed to a genuinely multipartite entangled state, a separable state can be obtained eventually. We investigate the minimal number of product states needed to be superposed so as to obtain a separable state in next section.

\section{3. Minimal number of superposed product states for separability}

It is impossible to vanish the entanglement of a bipartite state $|\psi\ra$ with Schmidt rank $r$ by superposing $r-2$ product states \cite{Halder}. Obviously, superposing $r-1$ product states is enough to make $|\psi\ra$ separable in general. However, if the superposed product states are required to be orthogonal, the things are different.

\begin{lemma}
For a bipartite pure state $|\psi\ra$ with Schmidt rank $r\geq 3$, there exits a set of $r$ mutually orthogonal product states $\mathbb{S}=\{|p_i\ra\}^{i=r}_{i=1}$ such that the superposition of $|\psi\ra$ and the product states gives a separable state.
\end{lemma}

\begin{proof}
$|\psi\ra$ has a Schmidt form,
$|\psi\ra=\sum^{r}_{i=1}a_i|i\ra|i\ra$ with $\sum^{r}_{i=1}a^2_i=1$. Set
\be
|p_i\ra=|i\ra\otimes\frac{1}{\sqrt{r-1}}
\sum_{j\neq i}|j\ra.\nonumber
\ee
Then $\mathbb{S}$ is a set of mutually orthogonal product states, $\la p_i|p_j\ra=\delta_{ij}$. Moreover, $\la\psi|p_i\ra=0,~\forall i=1,2,\ldots,r$. Taking $\lambda=\frac{1}{\sqrt{r}}$ and $\mu_i=\sqrt{\frac{r-1}{r}}a_i$ $(i=1,2,\ldots,r)$, we have the following separable superposed
state,
\be
\lambda|\psi\ra+
\sum^{r}_{i=1}\mu_i|p_i\ra
=\left(\sum^{r}_{i=1}a_i|i\ra\right)\otimes
\left(\frac{1}{\sqrt{r}}
\sum^{r}_{i=1}|i\ra\right).\nonumber
\ee
\end{proof}

Then it is an interesting issue to assure whether it is possible to obtain a separable state by superposing $r-1$ pairwise orthogonal product states with a given entangled state of Schmidt rank $r$. The answer depends on the detailed entangled state.

\begin{theorem}\label{thm4}
For a bipartite state with Schmidt form $|\psi\ra=\sum^r_{i=1}a_i|ii\ra$, it is possible to obtain a separable state by superposing $|\psi\ra$ with $r-1$ pairwise orthogonal product pure states only when not all $|a_i|\equiv 1/\sqrt{r}$ for $i=1,2,\cdots,r$.
\end{theorem}

\begin{proof} We first prove the case of $r=2$, i.e., $|\psi\ra=a_0|00\ra+a_1|11\ra$. Note that the orthogonality between $|\psi\ra$ and any separable state $|p_1\ra=|\alpha_1\ra|\beta_1\ra=(\cos\alpha_1|0\ra+e^{i\theta}\sin\alpha_1|1\ra)\otimes
(\cos\beta_1|0\ra+e^{i\delta}\sin\beta_1|1\ra)$ requires that
\be\label{con1}
a_0\cos\alpha_1\cos\beta_1+a_1\sin\alpha_1\sin\beta_1=0.
\ee
For some $p\in (0,1)$ the superposed state $\sqrt{p}|\psi_1\ra+\sqrt{1-p}|p_1\ra$ is separable only if
\be\label{con2}
pa_0a_1+\sqrt{p(1-p)}\Big[a_0\sin\alpha_1\sin\beta_1+a_1\cos\alpha_1\cos\beta_1\Big]=0.
\ee
(\ref{con1}) and (\ref{con2}) hold simultaneously only if $|a_0|\neq|a_1|$. Taking into account the normalization $a^2_0+a^2_1=1$, we have that the solutions of (\ref{con1}) and (\ref{con2}) exists for $p\in(0,1)$ only if $|a_0|\neq|a_1|\neq 1/\sqrt{2}$.

Next we prove the general case by induction. Without loss of generality, we assume that $|a_0|\neq|a_1|$. There exist $p\in(0,1)$ and $\alpha_1,\beta_1$ such that $\sqrt{p}\frac{1}{\sqrt{a^2_0+a^2_1}}(a_0|00\ra+a_1|11\ra)+\sqrt{1-p}
(\cos\alpha_1|0\ra+\sin\alpha_1|1\ra)\otimes(\cos\beta_1|0\ra+\sin\beta_1|1\ra)=
(x_0|0\ra+x_1|1\ra)\otimes(y_0|0\ra+y_1|1\ra)=|\phi_1\ra$. Therefore, the entanglement in the term $|00\ra+|11\ra$ is vanished. In this way, we can then vanish the entanglement given by $|\phi_1\ra$ and the next term $a_2|22\ra$ in $|\psi\ra$. By iterating $r-1$ times we obtain a separable state at last.
\end{proof}

{\it Example 5.}
Consider $|\psi\ra=\frac{3}{\sqrt{149}}|00\ra+\frac{6}{\sqrt{149}}|11\ra
+\frac{2\sqrt{26}}{\sqrt{149}}|22\ra$ with Schmidt rank 3. Set $|p_1\ra=\frac{1}{\sqrt{5}}(-|0\ra+2|1\ra)\otimes\frac{1}{\sqrt{17}}(4|0\ra+|1\ra)$. The entanglement given by the terms $|00\ra$ and $|11\ra$ is vanished by superposing $|p_1\ra$,
$$\begin{aligned}
\frac{3}{\sqrt{26}}&(\frac{1}{\sqrt{5}}|00\ra+\frac{2}{\sqrt{5}}|11\ra)+
\sqrt{\frac{17}{26}}|p_1\ra\\
&=\frac{1}{\sqrt{65}}(-|0\ra+8|1\ra)\otimes\frac{1}{\sqrt{2}}(|0\ra+|1\ra)
=|\alpha_1\ra\otimes|+\ra,
\end{aligned}
$$
where $|\alpha_1\ra=\frac{1}{\sqrt{65}}(-|0\ra+8|1\ra)$ and $|+\ra=\frac{1}{\sqrt{2}}(|0\ra+|1\ra)$, which are both orthogonal both to $|2\ra$. We then take $|p_2\ra=\frac{1}{\sqrt{5}}(|\alpha_1\ra+2|2\ra)
\otimes\frac{1}{\sqrt{21}}(4|+\ra-\sqrt{5}|2\ra)$. The superposition between $|\alpha_1\ra\otimes|+\ra$ and $|p_2\ra$ gives again a separable state,
$$\begin{aligned}
&\sqrt{\frac{3}{178}}\Big[\frac{\sqrt{5}}{3}|\alpha_1\ra|+\ra+\frac{2}{3}|22\ra\Big]
+\sqrt{\frac{175}{178}}|p_2\ra\\
&\quad=\frac{1}{\sqrt{89}}(5|\alpha_1\ra+8|2\ra)\otimes\frac{1}{\sqrt{6}}
(\sqrt{5}|+\ra-|2\ra).
\end{aligned}$$
Finally we have
$$
\begin{aligned}
&\lambda|\psi\ra+\mu_1|p_1\ra+\mu_2|p_2\ra\\
&=\frac{1}{\sqrt{89}}(5|\alpha_1\ra+8|2\ra)\otimes\frac{1}{\sqrt{6}}(\sqrt{5}|+\ra-|2\ra),
\end{aligned}
$$
where $\lambda=\sqrt{\frac{149}{178\times78}}$, $\mu_1=\sqrt{\frac{85}{178\times78}}$ and $\mu_2=\sqrt{\frac{175}{178}}$.

\bigskip
For tripartite systems, we consider the generalized GHZ states.

\begin{theorem}\label{thm5}
The entanglement of the genuinely multipartite entangled state $|\psi\ra=\sum^r_{i=1}a_i|iii\ra$ $(r\geq 3)$ can be vanished by superposing $2r-1$ product pure states which are orthogonal to $|\psi\ra$ (but not orthogonal mutually) when not all $|a_i|$ equal to $1/\sqrt{r}$.
\end{theorem}

\begin{proof}
Making use of the method in proving Lemma 2, we take
$|p_i\ra=|ii\ra\otimes\frac{1}{\sqrt{r-1}}\sum_{j\neq i}|j\ra$ which are orthogonal to $|\psi\ra$. Set $\lambda'=\frac{1}{\sqrt{r}}$ and $\mu'_i=\sqrt{\frac{r-1}{r}}a_i$. Then the superposition of the state $|\psi\ra$ and $|p_i\ra$ leads to a bi-separable state,
$$
\lambda'|\psi\ra+\sum^r_{i=1}\mu'_i|p_i\ra=\left(\sum^r_{i=1}a_i|ii\ra\right)
\otimes\left(\frac{1}{\sqrt{r}}\sum^r_{i=1}|i\ra\right).$$

According to \autoref{thm4}, it is possible to find nonzero $x$ and $y_i$ satisfying $x^2+\sum^{r-1}_{i=1}y^2_i=1$, and a set of orthogonal bipartite separable states $\{|q'_i\ra\}^{r-1}_{i=1}$ such that $$x\sum^r_{i=1}a_i|ii\ra+\sum^{r-1}_{i=1}y_i|q'_i\ra=|\phi\ra\otimes|\phi'\ra.$$
Set $\lambda=\lambda'\cdot x=\frac{x}{\sqrt{r}}$, $\mu_i=\mu'_i x=\sqrt{\frac{r-1}{r}}a_i x$ and $\eta_i=y_i$. Denote $|p_i\ra=|ii\ra\otimes\frac{1}{\sqrt{r-1}}\sum_{j\neq i}|j\ra$ and $|q_i\ra=|q'_i\ra\otimes\frac{1}{\sqrt{r}}\sum^r_{j=1}|j\ra$, $i=1,2,\ldots,r-1$, which are orthogonal to $|\psi\ra$. We have
\be\nonumber
\begin{aligned}
&\lambda|\psi\ra+\sum^r_{i=1}\mu_i|p_i\ra+\sum^{r-1}_{j=1}\eta_j|q_j\ra\\
&\quad=x\left(\lambda'|\psi\ra+\sum^r_{i=1}\mu'_i|p_i\ra\right)
+\sum^{r-1}_{j=1}y_j|q'_j\ra\otimes\frac{1}{\sqrt{r}}\sum^r_{j=1}|j\ra\\
&\quad=\left(x\sum^r_{i=1}a_i|ii\ra+\sum^{r-1}_{i=1}y_i|q'_i\ra\right)
\otimes\frac{1}{\sqrt{r}}\sum^r_{j=1}|j\ra\\
&\quad=|\phi\ra\otimes|\phi'\ra\otimes\frac{1}{\sqrt{r}}\sum^r_{j=1}|j\ra.
\end{aligned}
\ee
In this way we construct a set of $2r-1$ product states which are orthogonal to $|\psi\ra$ such that their superposition is a fully separable state.
\end{proof}

\section{4. Conclusions}

The superposition of pure states is the hypostatic cause of quantum entanglement. A suitable superposition of an entangled state with certain product states may give rise to a separable state. The multipartite entanglement is more complicated than bipartite cases. We have studied the robustness of genuine multipartite entanglement and inseparability of multipartite pure states under superposition with product pure states, by introducing the concepts of maximal and minimal Schmidt ranks for multipartite states. Based on the minimal and maximal Schmidt ranks of the first order, as well as the minimal Schmidt ranks of the second order, we have presented the number of (either arbitrary or orthogonal) product states superposed to a given (genuine) multipartite entangled state, so that the whole superposed state keeps (genuine) multipartite entangled. The results show also the resistance of GME and multipartite entanglement to the mixture of separable states. Unextendible product bases in multipartite quantum systems  attract a lot of attentions \cite{Bennett,Divincenzo,Xu,Halder2,Rout,Yuan,Bej,Zhou} while it supplies an effective way to construct bound entangled states. In \cite{Halder1} a set of non-orthogonal unextendible product bases is constructed by superposing a given entangled state and a set of product states in bipartite systems. This may also supply a method to construct unextendible product basis in multipartite systems.

\section{Acknowledgements}
This work is supported by Zhejiang Provincial Natural Science Foundation of China under Grant No. LZ23A010005, the National Natural Science Foundation of
China under grant Nos. 12171044 and 12075159, and the specific research fund of the Innovation Platform for Academicians of Hainan Province. H. Q. acknowledges the fellowship from the China scholarship council.



\end{document}